\documentclass[10pt, conference, compsocconf]{IEEEtran}
\usepackage[cmex10]{amsmath}

\usepackage{amsthm}
\newtheorem{definition}{Definition}
\newtheorem{lemma}[definition]{Lemma}
\newtheorem{theorem}[definition]{Theorem}
\newtheorem{corollary}[definition]{Corollary}
\theoremstyle{plain}
\usepackage{amssymb}
\usepackage{xspace}
\usepackage{xifthen}
\usepackage{xcolor}
\usepackage[noend]{algpseudocode}
\usepackage{algorithm}
\usepackage{tikz}
\usepackage{pgfplots}
\pgfplotsset{width=7.6cm,compat=1.9} 
\usetikzlibrary{arrows}
\usetikzlibrary{decorations.pathreplacing}

\newcommand{\pname}[0]{\textsc{Skueue}}

\newcommand{\enqueue}[1]{\textsc{Enqueue}\ensuremath{(#1)}\xspace}
\newcommand{\dequeue}{\textsc{Dequeue}\ensuremath{()}\xspace}
\newcommand{\push}[1]{\textsc{Push}\ensuremath{(#1)}\xspace}
\newcommand{\pop}{\textsc{Pop}\ensuremath{()}\xspace}
\newcommand{\join}[1]{\textsc{Join}\ensuremath{(#1)}\xspace}
\newcommand{\leave}{\textsc{Leave}\ensuremath{()}\xspace}

\newcommand{\enq}[2]{\textsc{Enq}\ensuremath{_{#1,#2}}\xspace}
\newcommand{\deq}[2]{\textsc{Deq}\ensuremath{_{#1,#2}}\xspace}
\newcommand{\op}[2]{\textsc{Op}\ensuremath{_{#1,#2}}\xspace}
\IEEEoverridecommandlockouts

\newcommand\copyrighttext{%
  \footnotesize \textcopyright 2018 IEEE. This is the full version of a corresponding paper published in the proceedings of IPDPS 2018 (DOI: 10.1109/IPDPS.2018.00113). 
  Personal use of this material is permitted.
  Permission from IEEE must be obtained for all other uses, in any current or future 
  media, including reprinting/republishing this material for advertising or promotional 
  purposes, creating new collective works, for resale or redistribution to servers or 
  lists, or reuse of any copyrighted component of this work in other works. 
}
\newcommand\copyrightnotice{%
\begin{tikzpicture}[remember picture,overlay]
\node[anchor=south,yshift=10pt] at (current page.south) {\fbox{\parbox{\dimexpr\textwidth-\fboxsep-\fboxrule\relax}{\copyrighttext}}};
\end{tikzpicture}%
}

\begin{document}
%
\title{\pname: A Scalable and Sequentially Consistent Distributed Queue*}


\author{\IEEEauthorblockN{Michael Feldmann, Christian Scheideler and Alexander Setzer\thanks{*This work was partially supported by the German Research Foundation (DFG) within the Collaborative Research Center ``On-The-Fly Computing'' (SFB 901)}}
\IEEEauthorblockA{ Department of Computer Science\\
Paderborn University, Germany\\
\{michael.feldmann, scheideler, alexander.setzer\}@upb.de} 
}

\maketitle
\copyrightnotice

\begin{abstract}
We propose a distributed protocol for a queue, called \pname{}, which spreads its data fairly onto multiple processes, avoiding bottlenecks in high throughput scenarios.
\pname{} can be used in highly dynamic environments, through the addition of \join{} and \leave{} requests to the standard queue operations \enqueue{} and \dequeue.
Furthermore \pname{} satisfies sequential consistency in the asynchronous message passing model.
Scalability is achieved by aggregating multiple requests to a batch, which can then be processed in a distributed fashion without hurting the queue semantics.
Operations in \pname{} need a logarithmic number of rounds w.h.p. until they are processed, even under a high rate of incoming requests.
\end{abstract}


\IEEEpeerreviewmaketitle

\section{Introduction} \label{sec:intro}
Like in the sequential world, efficient distributed data structures are important in order to realize efficient distributed applications. The most prominent type of distributed data structure is the distributed hash table (DHT). Many distributed data stores employ some form of DHT for lookup. Important applications include file sharing (e.g., BitTorrent), distributed file systems (e.g., PAST), publish subscribe systems (e.g., SCRIBE), and distributed databases (e.g., Apache Cassandra). Other distributed forms of well-known data structures, however, like queues, stacks, and heaps has been given much less attention though queues, for example, have a number of interesting applications as well. A distributed queue can be used to come up with a unique ordering of messages, transactions, or jobs, and it can be used to realize fair work stealing~\cite{DBLP:journals/jacm/BlumofeL99} since tasks available in the system would be fetched in FIFO order. Other applications are distributed mutual exclusion, distributed counting, or distributed implementations of synchronization primitives. Server-based approaches of realizing a queue in a distributed system already exist, like Apache ActiveMQ, IBM MQ, or JMS queues. Many other implementations of message and job queues can be found at http://queues.io/. However, none of these implementations provides a queue that allows massively parallel accesses without requiring powerful servers. The major problem of coming up with a fully distributed version of a queue is that its semantics are inherently sequential. Nevertheless, we are able to come up with a distributed protocol for a queue ensuring sequential consistency that fairly distributes the communication and storage load among all members of the distributed system and that can efficiently process even massive amounts of \enqueue{} and \dequeue requests. Our protocol works in the asynchronous message passing model and can also handle massive amounts of join and leave requests efficiently. We are not aware of any distributed queue with a comparable performance.

\subsection{Basic notation}

A \emph{Distributed Queue} provides four operations: \enqueue{}, \dequeue, \join{} and \leave. \enqueue{} adds an element to the queue and \dequeue removes an element from the queue so that the FIFO requirement is satisfied. \join{} allows a process to enter the system while \leave{} allows a process to leave the system. Let $\mathcal E$ be the universe of all elements that may possibly be put into the distributed queue.

While in a standard, sequential queue it is very easy to guarantee the FIFO property, it is much harder to guarantee in a distributed system, especially when messages have arbitrary finite delays and the processes do not have access to a local or global clock, as is usually assumed in the asynchronous message passing model. In essence, a global serialization of the requests has to be established without creating bottlenecks in the system. We will show that it is possible to obtain a serialization ensuring sequential consistency even under a high request rate. In order to define sequential consistency, we first need some notation.

Let \enq{v}{i} refer to the $i$-th \enqueue{} request that was called in process $v$.
Analogously, \deq{v}{i} refers to the $i$-th \dequeue request that was called in process $v$.
Furthermore, \op{v}{i} denotes the $i$-th (\enqueue{} or \dequeue) request that was called in process $v$.
We assume w.l.o.g. that every $e \in \mathcal E$ is enqueued at most once into the system (an easy way to achieve this is to make the calling process and the current count of requests performed a part of $e$).
Let $S$ be the set of all \enqueue{} and \dequeue requests issued by the processes in the system. We say that $\enq{v}{i}$ is {\em matched} with $\deq{w}{j}$ if the $\deq{w}{j}$ request returns the element contained in the $\enq{v}{i}$ request. 
Let $M$ be the set of all matchings.
Note that there may be requests that are not matched and thus not contained in $M$.

\begin{definition}\label{def:semantics}
A Distributed Queue protocol with operations \enqueue{} and \dequeue is \emph{sequentially consistent} if and only if there is an ordering $\prec$ on the set $S$ of all \enqueue{} and \dequeue requests issued to the system so that the set of all enqueue-dequeue matchings $M$ established by the protocol satisfies:
 \begin{enumerate}
  \item for all $(\enq{v}{i}, \deq{w}{j}) \in M$: $\enq{v}{i} \prec \deq{w}{j}$,
  \item for all $(\enq{v}{i}, \deq{w}{j}) \in M$: during the execution, there is no $\deq{u}{k}$ not contained in $M$ such that $\enq{v}{i} \prec \deq{u}{k} \prec \deq{w}{j}$, and there is no $\enq{u}{k}$ not contained in $M$ such that $\enq{u}{k} \prec \enq{v}{i} \prec \deq{w}{j}$,
  \item for all distinct $(\enq{u}{i}, \deq{v}{j}), (\enq{w}{k},\deq{x}{l}) \in M$ it does not hold:
    $\enq{u}{i} \prec \enq{w}{k} \prec \deq{x}{l} \prec \deq{v}{j}$ or $\enq{w}{k} \prec \enq{u}{i} \prec \deq{v}{j} \prec \deq{x}{l}$, and
  \item for all $v \in V$ and $i \in \mathbb{N}$: $\op{v}{i} \prec \op{v}{i+1}$
 \end{enumerate}
\end{definition}

Intuitively, the four properties have the following meaning:
The first property means that an element has to be enqueued before it can be dequeued.
The second property means that each \dequeue request returns a value if there is one in the queue and that each element passed as a parameter of an \enqueue{} request will be added to the queue.
The third property means elements are dequeued in the order they have been added to the queue.
Finally, the fourth property is the \emph{local consistency} property:
It means that for each single process, the requests performed by this process have to come up in $\prec$ in the order they were executed by that process.

Note that if there is only a single process in the system, then the \enqueue{} and \dequeue operations of the Distributed Queue have exactly the same semantics as a classical queue.

\subsection{Model}

The distributed queue consists of multiple processes that are interconnected by some overlay network.
We model the overlay network as a directed graph $G = (V, E)$, where $V$ represents the set of processes and an edge $(v,w)$ indicates that $v$ knows $w$ and can therefore send messages to $w$.
Each process $v$ can be identified by a unique identifier $v.id \in \mathbb{N}$.

We consider the asynchronous message passing model where every process $v$ has a set $v.Ch$ for all incoming messages called its \emph{channel}. That is, if a process $u$ sends a message $m$ to process $v$, then $m$ is put into $v.Ch$. A channel can hold an arbitrary finite number of messages and messages never get duplicated or lost.

Processes may execute \textit{actions}: An action is just a standard procedure that consists of a name, a (possibly empty) set of parameters, and a sequence of statements that are executed when calling that action.
It may be called locally or remotely, i.e., every message that is sent to a process contains the name and the parameters of the action to be called. We will only consider messages that are remote action calls. An action in a process $v$ is {\em enabled} if there is a request for calling it in $v.Ch$. Once the request is processed, it is removed from $v.Ch$. We assume fair message receipt, i.e., every request in a channel is eventually processed.
Additionally, there is an action that is not triggered by messages but is executed periodically by each process.
We call this action \textsc{Timeout}.

We define the \textit{system state} to be an assignment of a value to all protocol-specific variables in the processes and a set of messages to each channel. A \textit{computation} is a potentially infinite sequence of system states, where the state $s_{i+1}$ can be reached from its previous state $s_i$ by executing an action that is enabled in $s_i$.

We place no bounds on the message propagation delay or the relative process execution speed, i.e., we allow fully asynchronous computations and non-FIFO message delivery.

For the runtime analysis, we assume the standard synchronous message passing model, where time proceeds in \emph{rounds} and all messages that are sent out in round $i$ will be processed in round $i+1$.
Additionally, we assume that each process executes its \textsc{Timeout} action once in each round.

\subsection{Related Work}
The most important type of distributed data structure is the distributed hash table, for which seminal work has been done by Plaxton et al.~\cite{DBLP:conf/spaa/PlaxtonRR97} and Karger et al.~\cite{DBLP:conf/stoc/KargerLLPLL97}.
Distributed hash tables have a wide range of practical realizations, such as Chord~\cite{DBLP:conf/sigcomm/StoicaMKKB01}, Pastry~\cite{DBLP:conf/middleware/RowstronD01}, Tapestry~\cite{DBLP:journals/jsac/ZhaoHSRJK04} or Cassandra~\cite{DBLP:conf/podc/LakshmanM09}.
Our queue protocol makes use of a distributed hash table through consistent hashing.

Distributed hash tables do not support range queries, so distributed trees were proposed, e.g. in~\cite{DBLP:conf/isaac/AlaeiTG05, DBLP:conf/sigmod/KrollW94}, to overcome this.

There is a wealth of literature on \emph{concurrent} data structures.
Consider, for example, \cite{DBLP:conf/podc/MichaelS96} for a queue,~\cite{DBLP:journals/jpdc/HendlerSY10} for a stack,~\cite{DBLP:conf/ipps/ShavitL00} for a priority queue or~\cite{DBLP:reference/crc/MoirS04} for a general survey.
These structures allow multiple processes to send requests to a data structure that is stored in shared memory.
Hendler et al.~\cite{DBLP:conf/wdag/HendlerIST10} present a scalable synchronous concurrent queue, where they used a parallel flat-combining algorithm similar to the aggregation technique used in this work:
A single 'combiner' thread gets to know requests of other threads and then executes these requests on the queue.
However, they do not provide any guarantees on the semantics, as their queue is considered to be \emph{unfair}, meaning that it does not impose an order on the servicing of requests.
Shavit and Taubenfeld formulated some (relaxed) semantics for concurrent queues and stacks in~\cite{DBLP:journals/dc/ShavitT16}.
The main difference of concurrent data structures compared to distributed data structures is that there has to be a single instance that stores the data, whereas distributed data structures are fully decentralized.

A scalable distributed heap called \emph{SHELL} has been presented by Scheideler and Schmid in~\cite{DBLP:conf/icalp/ScheidelerS09}. \emph{SHELL}'s topology resembles the De Bruijn graph and is shown to be very resilient against Sybil attacks.
Our protocol uses the virtual De Bruijn graph from Richa et al.~\cite{DBLP:conf/sss/RichaSS11}, which is based on~\cite{DBLP:journals/talg/NaorW07}, where Naor and Wieder showed how to construct P2P systems in the continuous space.

Plenty of work has also been done on \emph{distributed queuing}, but this is very different from our approach.
Distributed queuing is all about the participants of the system forming a queue: Every process introduces itself to its predecessor and (depending on its position) knows its successor in the queue.
Distributed queuing is not about inserting elements into a distributed data structure that is maintained by multiple processes, which can generate requests to the data structure.
See, for example, the Arrow protocol in~\cite{DBLP:conf/podc/HerlihyTW01}, which was made self-stabilizing in~\cite{DBLP:journals/tpds/TirthapuraH06}, or a protocol for dynamic networks in~\cite{DBLP:journals/ppl/SharmaB15}.

\subsection{Our Contribution}
We propose a protocol for a distributed queue which guarantees sequential consistency (Definition~\ref{def:semantics}).
Requests can be handled very effectively due to the aggregation of multiple requests to a batch.
This fact makes our queue highly scalable for both, a large number of processes and a high load of queue requests.
More precisely, when assuming synchronous message passing, our \enqueue{} and \dequeue operations are processed in $\mathcal O(\log n)$ rounds w.h.p.
Furthermore we show that we can process $n$ \join{} or $n/2$ \leave operations in $\mathcal O(\log n)$ rounds.
Through the usage of a distributed hash table, our distributed queue allocates its elements equally among all processes, such that no process stores significantly more elements than the rest.

The paper is structured as follows: In Section~\ref{sec:pre} we describe the linearized De Bruijn network topology, into which we embed a distributed hash table.
The general ideas for our protocol are presented in Section~\ref{sec:enq_deq} along with descriptions for \enqueue{} and \dequeue operations.
In Section~\ref{sec:join_leave} we extend the protocol in order to support \join{} and  \leave operations.
We explain how to modify \pname{} in order to work as a distributed stack and present experimental results for both the queue and the stack (Sections~\ref{sec:stack} and~\ref{sec:evaluation}).
Before we conclude the paper in Section~\ref{sec:conclusion}, we analyze the most important properties of our protocol in Section~\ref{sec:analysis}.

\section{Preliminaries} \label{sec:pre}
\subsection{Linearized De Bruijn Network}
We adapt a dynamic version of the De Bruijn graph from~\cite{DBLP:conf/sss/RichaSS11}, which is based on~\cite{DBLP:journals/talg/NaorW07}, for our network topology:

\begin{definition}
	The \emph{Linearized De Bruijn network} (LDB) is a directed graph $G = (V, E)$, where each process $v$ emulates $3$ (virtual) nodes: A \emph{left virtual node} $l(v) \in V$, a \emph{middle virtual node} $m(v) \in V$ and a \emph{right virtual node} $r(v) \in V$.
	The middle virtual node $m(v)$ has a real-valued \emph{label}\footnote{We may indistinctively use $v$ to denote a node or its label, when clear from the context.} in the interval $[0,1)$.
	The label of $l(v)$ is defined as $m(v)/2$ and the label of $r(v)$ is defined as $(m(v)+1)/2$.
	The collection of all virtual nodes $v \in V$ is arranged in a sorted cycle ordered by node labels, and $(v,w) \in E$ if and only if $v$ and $w$ are consecutive in this ordering (\emph{linear edges}) or $v$ and $w$ are emulated by the same process (\emph{virtual edges}).
\end{definition}

We will assume that the label of a middle node $m(v)$ is determined by applying a publicly known pseudorandom hash function on the identifier $v.id$.
We say that a node $v$ is \emph{right} (resp. \emph{left}) of a node $w$ if the label of $v$ is greater (resp. smaller) than the label of $w$, i.e., $v > w$ (resp. $v < w$). 
If $v$ and $w$ are consecutive in the linear ordering and $v < w$ (resp. $v > w$), we say that $w$ is $v$'s \emph{successor} (resp. \emph{predecessor}) and denote it by $succ(v)$ (resp. $pred(v)$).
As a special case we define $pred(v_{min}) = v_{max}$ and $succ(v_{max}) = v_{min}$, where $v_{min}$ is the node with minimal label value and $v_{max}$ is the node with maximal label value.
This guarantees that each node has a well defined predecessor and successor on the sorted cycle.
More precisely, each node $v$ maintains two variables $pred(v)$ and $succ(v)$ for storing its predecessor and successor nodes.
Whenever a node $v$ gets to know the reference of another node $w$, such that $w$ is stored in either $pred(v)$ or $succ(v)$, we assume that $v$ also gets to know whether $w$ is a left, middle or right virtual node.
This can be done easily by attaching the information to the message that contains the node reference.
By adopting the result from~\cite{DBLP:conf/sss/RichaSS11}, one can show that routing in the LDB can be done in $\mathcal O(\log n)$ rounds w.h.p.:

\begin{lemma}\label{lemma:LDB:routing}
For any $p \in [0,1)$, routing a message from a source node $v$ to a node that is the predecessor of $p$ in the LDB can be done in $\mathcal O(\log n)$ rounds w.h.p.
\end{lemma}

\subsection{Distributed Hash Table} \label{sec:pre:dht}
In order to store the elements of our queue in a distributed fashion, we use a distributed hash table (DHT) that makes use of consistent hashing: Elements $e \in \mathcal E$ that should be stored in the DHT will be assigned a unique position $p(e) \in \mathbb{N}_0$ by \pname{}.
This position can then be hashed to a real-valued key $k(p(e)) \in [0,1)$ via a publicly known pseudorandom hash function.
A node $v$ is responsible for storing all elements whose keys are within the interval $[v,succ(v))$.
Thus, if we want to insert (resp. delete) an element $e \in \mathcal E$, we only have to search for the node $v$ with $v \leq k(p(e)) < succ(v)$ and tell $v$ to store $e$.
The search for $v$ can be performed in $\mathcal O(\log n)$ rounds according to Lemma~\ref{lemma:LDB:routing}.
We will use the following operations in \pname{}:

\begin{enumerate}
	\item \textsc{Put}($e$, $k$) Inserts the element $e \in \mathcal E$ with key $k$ into the DHT.
	\item \textsc{Get}($k$, $v$): Removes the element with key $k$ from the DHT and delivers it to the initiator $v$ of the request.
\end{enumerate}

It is well known for consistent hashing that it is \emph{fair}, meaning that each node stores the same amount of elements for the DHT on expectation.

\begin{lemma}\label{lemma:dht:fairness}
	Consistent Hashing is fair.
\end{lemma}

\section{Enqueue \& Dequeue} \label{sec:enq_deq}
Throughout this paper, a \emph{queue operation} is either an \enqueue{} or a \dequeue request.

The main challenge to guarantee the sequential consistency from Definition~\ref{def:semantics} lies in the fact that messages may outrun each other, since we allow fully asynchronous computations and non-FIFO message delivery.
In a synchronous environment, this would not be a problem.
Another problem we have to solve is that the rate at which nodes issue queue requests may be very high.
As long as we process each single request one by one, scalability cannot be guaranteed.

The general idea behind \pname{} is the following: First, we aggregate batches of queue operations to the leftmost node in the LDB, called \emph{anchor}, by forwarding them to the leftmost neighbor at each hop.
By doing so, every involved node implicitly becomes part of an \emph{aggregation tree}.
The anchor then assigns a position $p \in \mathbb{N}_0$ in the DHT for each queue operation and spreads all positions for the queue operations over the aggregation tree such that sequential consistency (Definition~\ref{def:semantics}) is fulfilled.
Nodes in the aggregation tree then generate \textsc{Put} and \textsc{Get} requests for the respective positions in the DHT.
We describe this approach in more detail now.

\subsection{Operation Batch}
Whenever a node initiates a queue operation, it has to buffer it in its local storage.
We are going to represent the sequence of buffered queue operations by a \emph{batch}:

\begin{definition}[Batch] \label{def:batch}
	A \emph{batch} $B$ (of queue operations) is a sequence $(op_1,\ldots,op_k) \in \mathbb{N}_0^k$, for which it holds that for all odd $i$, $1 \leq i \leq k$, $op_i$ represents the length of the $i$-th enqueue sequence.
	Similarly, for all even $i, 1 < i \leq k$, $op_i$ represents the length of the $i$-th dequeue sequence.
	Denote the batch $(0)$ as \emph{empty}.
\end{definition} 

We are able to \emph{combine} two batches $(op_1,\ldots,op_k)$ and $(op_1',\ldots,op_l')$ by computing $B = (op_1'',\ldots,op_m'')$ with $op_i'' = op_i + op_i'$ and $m = \max\{k,l\}$ (we define $op_i = 0$ if $i > k$ and $op_i' = 0$ if $i > l$).
If a batch $B$ is the combination of batches $A_1,\ldots,A_k$, then we denote $A_1,\ldots,A_k$ as \emph{sub-batches}.
Each node may store two types of batches locally: One batch that is currently being processed and another batch that waits for being processed and acts as the buffer for newly generated queue operations.
For a node $v$, we call the former batch $v.B$ and the latter one $v.W$.
We denote $v$ as the \emph{owner} of the batch $v.B$.

Whenever a node $v$ generates a queue operation $op$, we update the batch $v.W = (op_1,\ldots,op_k)$ in the following way: If $op$ is an \enqueue{} request, then we increment $op_k$ if $k$ is odd, otherwise we add a $1$ to the batch by setting $v.W = (op_1,\ldots,op_k, 1)$.
Similarly, if $op$ is a \dequeue request, we increment $op_k$, if $k$ is even, otherwise we set $v.W = (op_1,\ldots,op_k, 1)$.
By doing so, the batch $v.W$ respects the local order in which queue operations are generated by $v$, which is important for guaranteeing sequential consistency.

\subsection{Aggregation Tree}
All (virtual) nodes in the LDB implicitly form an aggregation tree.
In order to do this, a node $v$ needs to know both, its parent and its child nodes in the tree.
Both depend on whether $v$ is a left, middle or right virtual node (see Figure~\ref{fig:aggregation_tree} for an example).

The parent node $p(v)$ of some node $v$ in the aggregation tree is always the node that is $v$'s leftmost neighbor.
More specifically, if $v$ is a middle virtual node, then $p(v) = l(v)$.
If $v$ is a left virtual node then $p(v) = pred(v)$.
Finally, if $v$ is a right virtual node, then $p(v) = m(v)$.

Next, we describe how a node $v$ knows its child nodes (denoted by the set $C(v)$) in the aggregation tree, assuming that the node set is static (we describe how to handle \join{} and \leave requests in Section~\ref{sec:join_leave}).
If $v$ is a middle virtual node, then either $C(v) = \{r(v), succ(v)\}$ (if $succ(v)$ is a left virtual node) or $C(v) = \{r(v)\}$ (otherwise).
If $v$ is a left virtual node, then either $C(v) = \{m(v), succ(v)\}$ (if $succ(v)$ is a left virtual node) or $C(v) = \{m(v)\}$ (otherwise).
Last, if $v$ is a right virtual node, then $C(v) = \emptyset$.
Intuitively, each node has its next virtual node as a child and also its successor if that successor is a left node.
A right virtual node cannot have a left virtual node as a right neighbor (as the id of a right virtual node is always at least $0.5$ and the id of a left virtual node is always less than $0.5$).

\begin{figure}[ht]
 	\centering
 	\begin{tikzpicture}[main node/.style={circle,draw,align=center, minimum size=0.3cm, inner sep=0pt}]	
     	\node[main node, label={[yshift=-0.85cm]$l(u)$}] (A) at (0,0) {};
     	\node[main node, label={[yshift=-0.85cm]$l(v)$}] (B) at (1.5,0) {};
     	\node[main node, label={[yshift=-0.85cm]$m(u)$}] (C) at (3,0) {};
     	\node[main node, label={[yshift=-0.85cm]$m(v)$}] (D) at (4.5,0) {};
     	\node[main node, label={[yshift=-0.85cm]$r(u)$}] (E) at (6,0) {};
     	\node[main node, label={[yshift=-0.85cm]$r(v)$}] (F) at (7.5,0) {};
     		
     	\draw[<->, line width=1.5pt] (A) to (B);
     	\draw[<->, line width=0.5pt] (B) to (C);
     	\draw[<->, line width=0.5pt] (C) to (D);
     	\draw[<->, line width=0.5pt] (D) to (E);
     	\draw[<->, line width=0.5pt] (E) to (F);
     	\draw[<->, line width=1.5pt, bend left = 60] (A) to (C);
     	\draw[<->, line width=1.5pt, bend left = 60] (B) to (D);
     	\draw[<->, line width=1.5pt, bend left = 60] (C) to (E);
     	\draw[<->, line width=1.5pt, bend left = 60] (D) to (F);
 	\end{tikzpicture}
 	\caption{A LDB consisting of $6$ nodes (corresponding to $2$ processes $u$ and $v$). Bold linear/virtual edges define the corresponding aggregation tree.}
 	\label{fig:aggregation_tree}
\end{figure}
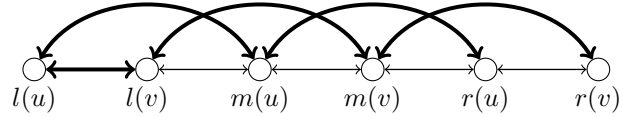

Observe that nodes are able to find their connections in the tree by relying on local information only.
Thus, for the rest of the paper we assume that every node knows its parent and child nodes in the aggregation tree at any time.

From Lemma~\ref{lemma:LDB:routing}, we directly obtain an upper bound for the height of the aggregation tree:

\begin{corollary} \label{cor:tree_height}
	The aggregation tree based on the LDB has height $\mathcal O(\log n)$ w.h.p.
\end{corollary}

We are now ready to describe our approach for queue operations in detail, dividing it into $4$ stages.

\subsection{Stage 1: Aggregating Batches}
Every time a node $v$ calls its \textsc{Timeout} (see Algorithm~\ref{algo:phase_1}) method, it checks whether its batch $v.B$ is empty and its batch $v.W$ contains the batches of all of its child nodes in the aggregation tree.
If that is the case, then $v$ transfers the data of $v.W$ to $v.B$ and sends out a message containing the contents of $v.B$ to $p(v)$.
Additionally, $v$ memorizes the sub-batches that are combined in $v.B$ such that it can determine the child node that sent the sub-batch to $v$.
We proceed this way in a recursive manner, until the root node $v_0$ of the aggregation tree, denoted as \emph{anchor} from now on, has received all batches from its child nodes.
Then it combines these batches with its own batch $v_0.W$ into $v_0.B$ and switches to the next stage by locally calling \textsc{Assign} (see Algorithm~\ref{algo:phase_2_4}).

\newcommand{\pushcode}[1][1]{\hskip\dimexpr#1\algorithmicindent\relax}

\begin{algorithm}[ht]
\caption{Stage 1 \Comment{Executed by each node $v$}}
\label{algo:phase_1}
\begin{algorithmic}[1]
\Procedure{Timeout}{}
	\If{$v.B = (0) \wedge v.W$ contains sub-batches from all \\ 
	\pushcode[0] $c \in C(v)$}
		\State $v.B \gets v.W$
		\State $v.W \gets (0)$
		\If{$v$ is the anchor node $v_0$}
			\State \textsc{Assign}($v_0.B$) \Comment{Switch to Stage 2}
		\Else
			\State $p(v) \gets$ \Call{Aggregate}{$v.B$}
		\EndIf
	\EndIf
\EndProcedure
\Procedure{Aggregate}{$B$}
	\State $v.W \gets v.W \cup B$
\EndProcedure
\end{algorithmic}
\end{algorithm}

\subsection{Stage 2: Assigning Positions}
At the anchor $v_0$ we maintain two variables $v_0.first \in \mathbb{N}_0$ and $v_0.last \in \mathbb{N}_0$, such that the invariant $v_0.first \leq v_0.last + 1$ holds at any time.
The interval $[v_0.first, v_0.last]$ represents the positions that are currently occupied by elements of the queue, which implies that the current size of the queue is equal to $v_0.last - v_0.first + 1$.

Now we describe how the anchor processes its batch $v_0.B = (op_1,\ldots,op_k)$ at the start of this stage.
Based on its variables $v_0.first, v_0.last$, $v_0$ computes intervals $[x_1,y_1],\ldots,[x_k,y_k]$ by processing each element in the batch $(op_1,\ldots,op_k)$ in ascending order of their indices $i$.
If $i$ is odd, then $v_0$ sets the interval $[x_i,y_i]$ to $[v_0.last + 1, v_0.last + op_i]$ and increases $v_0.last$ by $op_i$ afterwards.
Similarly, if $i$ is even, then $v_0$ sets the interval $[x_i,y_i]$ to $[v_0.first, \min\{v_0.first + op_i - 1, v_0.last\}]$ and updates $v_0.first$ to $\min\{v_0.first + op_i, v_0.last + 1\}$ afterwards.
By doing so, we assigned an interval to each sequence $op_i$ of requests, implying that we can assign a position to each single queue operation of such a sequence (which is part of the next stage).
Note that in case the queue is empty or does not hold sufficiently many elements and the anchor has to assign positions to some sequence of \dequeue requests of length $k$, it either holds $x_i = y_i + 1$ (if the queue is empty) or $x_i - y_i < k$ for the computed interval $[x_i, y_i]$.

\subsection{Stage 3: Decomposing Position Intervals}
Once $v_0$ has computed all required position intervals $[x_1,y_1],\ldots,[x_k,y_k]$ for a batch, it starts broadcasting these intervals over the aggregation tree, by calling \textsc{Serve} on its child nodes, see Algorithm~\ref{algo:phase_2_4}.
When a node $v$ in the tree receives a collection $[x_1,y_1],\ldots,[x_{k'},y_{k'}]$ of intervals, it decomposes the intervals with respect to each sub-batch $B_1,\ldots,B_l$ of $v.B$ (recall that $v$ has memorized this combination).
Consider a sub-batch $B_i = (op_1,\ldots,op_m)$ of $v.B$.
We describe how $v$ is able to assign a (sub-)interval to each $op_i$.
Assume $i$ is odd for $op_i$ (corresponding to $op_i$ many \enqueue{} requests). 
Then $v$ assigns the (sub\mbox{-})interval $[x_i, x_i + op_i - 1]$ to $op_i$.
Afterwards we update $[x_i,y_i]$ by setting $[x_i,y_i] = [x_i+op_i,y_i]$.
This implies that every \enqueue{} request is assigned a unique position.

Now assume $i$ is even for $op_i$ (corresponding to $op_i$ many \dequeue requests).
Then $v$ assigns the (sub\mbox{-})interval $[x_i, \min\{x_i+op_i-1, y_i\}]$ to $op_i$.
Afterwards we set $[x_i,y_i] = [\min\{x_i+op_i, y_i+1\},y_i]$.
This implies that \dequeue requests are either assigned a position or immediately return $\perp$ in case the interval is not large enough to assign a position to all \dequeue requests.

Once each sub-batch of $v.B$ has been assigned to a collection of (sub\mbox{-})intervals, we send out these intervals to the respective child nodes in $C(v)$.
Applying this procedure in a recursive manner down the aggregation tree yields an assignment of a position to all \enqueue{} and \dequeue requests.

\subsection{Stage 4: Updating the DHT}
Now that a node $v$ knows the exact position $p \in \mathbb{N}_0$ for each of its queue operations, it starts generating \textsc{Put} and \textsc{Get} requests.
For an request \enqueue{e} that got assigned to position $p$, $v$ issues a \textsc{Put}($e$, $k(p)$) request to insert $e$ into the DHT (recall that the key $k(p) \in [0,1)$ is just the real-valued hash of $p$).
This finishes the \enqueue{e} request.
For a \dequeue request that got assigned to position $p$, $v$ issues a \textsc{Get}($k(p)$, $v$) request.
Since in the asynchronous message passing model, it may happen that a \textsc{Get} request arrives at the correct node in the DHT \emph{before} the corresponding \textsc{Put} request, each \textsc{Get} request waits at the node responsible for the position $k$ until the corresponding \textsc{Put} request has arrived.
This is guaranteed to happen, as we do not consider message loss.

Once a node has sent out all its DHT requests, it switches again to Stage 1, in order to process the next queue operations.

\begin{algorithm}[ht]
\caption{Stages 2-4}
\label{algo:phase_2_4}
\begin{algorithmic}[1]
\Procedure{Assign}{$B$} \Comment{Executed by the anchor}
	\State Compute intervals $I = [x_1,y_1],\ldots,[x_k,y_k]$ from $B$
	\State \textsc{Serve}($I$) \Comment{Switch to Stage 3}
\EndProcedure
\Procedure{Serve}{$I$} \Comment{Executed by each node $v$}
	\State Decompose $I$ depending on $C(v)$ and $v.B$
	\ForAll{$c \in C(v)$}
		\State Forward sub-intervals $I_c \subset I$ to $c$ via \textsc{Serve}($I_c$)
	\EndFor
	\State Forward \textsc{Put}/\textsc{Get} requests to the DHT
	\State $v.B \gets (0)$ \Comment{Return to Stage 1}
\EndProcedure
\end{algorithmic}
\end{algorithm}

We defer the analysis of the \enqueue{} and \dequeue requests to Section~\ref{sec:analysis}.

\section{Join \& Leave} \label{sec:join_leave}
When a process enters or leaves the system, this entails several changes to the system in order to get into the state assumed in Section~\ref{sec:enq_deq}:
The DHT has to be updated, which includes movement of data to joining or from leaving nodes, the LDB has to be updated and meanwhile the aggregation tree changes.
To prevent chaos caused by the latter, we handle joins and leaves \emph{lazily}.
This means that a node $v$ joining or leaving the network will be assigned a node $u$ \emph{responsible for} $v$.
$u$ then acts as a representative for $v$ meaning that $u$ takes over $v$'s DHT data and emulates $v$ in the case of $v$ being a leaving node, or relays $v$'s \enqueue{} or \dequeue{} requests in the case of $v$ being a joining node.
Only after a sufficiently large number of nodes has requested to join or leave the system (which is counted at the anchor), the system enters a special state in which no further batches are sent out.
During this state, joining nodes are fully integrated into the system (meaning they do no longer need a node responsible for them) and nodes that left can end being emulated.
In the following, we will specify the details of this.
Keep in mind that a node that requested to join the system and that is not yet fully integrated into the system is called a \emph{joining node} and a node that requested to leave the system and that has not yet left is called a \emph{leaving} node.

Note that if a process $v$ wants to join or leave the network, we have to integrate or disconnect the three nodes $l(v), m(v), r(v) \in V$ into or from the system.
Therefore, we generate a \join{} or \leave{} request for each of these three nodes separately.
In the following we describe how one of these requests is handled.

\subsection{Join} \label{subsec:join}
Assume a node $v$ wants to join the system and further assume $v > v_0$ for now (we will consider the other case separately below).
Then it sends a \join{v} request to a node $w$.
We assume that if node $v$ wants to join the system via \join{v} at node $w$, we route $v$ from $w$ to the node $u$ such that $u < v < succ(u)$ or $succ(u) < u < v$ (in case the edge $(u,succ(u))$ closes the cycle) holds.
We define $u$ to be \emph{responsible for} \join{v}.
$u$ has the following tasks: 
First, it introduces itself to $v$.
Second, it hands over to $v$ all DHT data whose key is in $v$'s interval.
Any \textsc{Put} or \textsc{Get} requests for data with keys in this interval $u$ will forward to $v$ from then on.
Third, $u$ considers $v$ to be a child in its aggregation tree, meaning that $v$ is able to send \enqueue{} or \dequeue{} requests via $u$.
Fourth, $u$ notifies the anchor that there is an additional node that has joined the system.
For this, we extend the notion of a batch $B$ from Definition~\ref{def:batch}, such that it stores an additional number $B.j \in \mathbb{N}_0$ representing the number of \join{} requests that $u$ is responsible for.
Node $u$ proceeds in the same manner as for the queue operations in Section~\ref{sec:enq_deq}: 
It buffers the request in $u.W$ by adding $1$ to $u.W.j$ and once $u.B$ is empty and $u$ has received batches from every child, $u$ transfers all \enqueue{}, \dequeue and \join{} requests stored in $u.W$ to $u.B$ forwards the batch up in the aggregation tree.
Any intermediate node, when combining batches $B_1, \dots B_k$, calculates the sum of the $B_i.j$ values for the combined batch.
This way the anchor learns a lower bound on the total number of joining nodes (note that additional nodes may have requested to join but knowledge of this has not yet reached the anchor).

Note that a node $u$ may become responsible for several joining nodes $v_1, \dots v_k$.
In this case, everything written before still holds with one exception:
Assume $u$ is responsible for nodes $v_1, \dots, v_k$ and becomes responsible for an additional node $v'$ such that a node $v_i$ is the closest predecessor of $v'$.
Then $u$ does not transfer the DHT from itself to $v'$ but issues $v_i$ to transfer the DHT data to $v'$ and sends a reference of $v'$ to $v_i$.
Using this reference, $v_i$ can forward any \textsc{Put} or \textsc{Get} requests that fall within the remit of $v'$.

If the anchor can observe that the number of joining nodes exceeds the number of successfully integrated nodes when processing a batch, it sends the computed intervals down the aggregation tree as usual (c.f.~Section~\ref{sec:enq_deq}), but attaches a flag to the message indicating that the \emph{update phase} should be entered (thus informing all nodes of this).
In this phase, no node will send out a new batch until it has been informed that the update phase is over.
Instead, nodes responsible for other nodes will fully integrate these nodes into the system.
This works in the following way: When a node $u \neq v_0$ in the aggregation tree receives the intervals from its parent node, it proceeds as described in Section~\ref{sec:enq_deq}, i.e., it splits the intervals, forwards intervals to its children and possibly sends out \textsc{Put} and \textsc{Get} requests.
Additionally, $u$ stores the parent $p_{old}(u)$ in the aggregation tree it received the intervals from and all children $C_{old}(u)$ it forwards the intervals to.
This is required because in the update phase the aggregation trees may change, but the acknowledgments that the joining nodes have been integrated successfully need to be aggregated via the old aggregation tree.
That means that as soon as $u$ has integrated all nodes it is responsible for (if any) and received acknowledgments from all nodes in $C_{old}(u)$ (if any), it sends an acknowledgment to $p_{old}(u)$ and forgets $C_{old}(u)$ and $p_{old}(u)$.
$v_0$ behaves similar to any other node $u$, i.e., it also stores its old children, processes \textsc{Put} and \textsc{Get} requests and also starts integrating nodes it is responsible for.
However, when it has finished in doing so, and received all acknowledgments from the nodes in $C_{old}(v_0)$, it propagates down in the new aggregation tree a message indicating that the update phase is over (note that we consider the case of a joining node to the left of the anchor below).
This is safe because it can be shown by induction that when $v_0$ has received acknowledgments from all its children, every node in the tree has finished integrating at least all joining nodes that were joining when the anchor entered the update phase.
Once a node has received an indication that the update phase is over, it starts aggregating and sending out batches again.
We now describe how integrating a joining node works.

Consider a node $u$ that is responsible for $v_1, \dots, v_k$.
W.l.o.g., we assume $u < v_1 < \ldots < v_k < succ(u)$.
$u$ introduces $v_i$ to $v_{i+1}$ and vice versa for all $i \in \{1,\ldots,k-1\}$ and introduces $succ(u)$ to $v_k$ and vice versa.
Finally, $u$ drops its connections to $v_2, \dots, v_k$ and $succ(u)$.
%
%

Note that the nodes $v_i$ already stored their corresponding DHT data from the point when $u$ became responsible for them.
Due to changes in the De Bruijn graph it may happen that \textsc{Put} or \textsc{Get} requests do not need to be routed to the same target as before.
However, if a \textsc{Put} request is at a node $v$ that is not responsible for storing the corresponding element $e$, $v$ must have a neighbor that is closer to the node responsible for storing $e$.
This is because whenever $v$ removes an edge to a neighbor during join, it has learned to know a closer one in the same direction before.
Thus $v$ can forward it into the right direction.
Similarly, if a \textsc{Get} request is at a node $v$ that does not store the desired element $e$, $v$ can wait until it either stores $e$ or until it has learned to know a node that is closer to the target than itself.
Since eventually our procedure forms the correct De Bruijn topology, these requests will be answered.

\paragraph{Updating the Anchor} \label{sec:join:anchor_update}
We now consider the special case, where at least one new node $v$'s label is smaller than the label of the current anchor $v_0$.
Then the node responsible for $v$ is the node $u$ with maximum label, i.e., $u = pred(v_0)$.
$u$ behaves as described before.
However, when $v_0$ has received all acknowledgments from it children and integrated the nodes it is responsible for, it does not send out the message indicating that the update phase is over (note that $v_0$ can determine that a node $v < v_0$ has joined because its neighborhood to the left has changed).
Instead, $v_0$ searches for the leftmost node $v_0'$ and and then transfers its interval $[v_0.first, v_0.last]$ to $v_0'$.
From that point on, $v_0'$ will behave as the new anchor and send the message indicating that the update phase is over down in the new aggregation tree.

\subsection{Leave} \label{subsec:leave}
The general strategy for leaves is the following:
For each leaving node $v$, the process emulating the left neighbor $u$ of $v$ creates a virtual node $v'$ that acts as a replacement for $v$, i.e., $v'$ will store $v$'s DHT data, be responsible for the nodes $v$ was responsible for and have the same connections as $v$ had.
As soon as this replacement has been created, the corresponding edges have been established, the edges to $v$ have been removed, and all messages on their way to $v$ have been delivered and successfully forwarded from $v$, $v$ is safe to leave the system and does so.
The challenge is to deal with neighboring leaving nodes: 
If $v$ has a neighbor that is also leaving, then this neighbor does not want to establish a new edge, which might result in a deadlock situation.
Thus, we have to prioritize leaves: 
Whenever two neighboring nodes $u$ and $v$ determine that they both want to leave, the one with the higher identifier postpones its attempt to leave until the other one has left the system.
Since in any case there is a unique leftmost leaving node, there will always be a node that can leave the system, which inductively yields that all nodes eventually leave.
To enable this, each node that calls \leave{} first asks all its left neighbors if it is allowed to do so.
Only if all of them acknowledge, it starts the actual procedure to leave.
Note that a node $u$ that acknowledged a right neighbor $v$ that it may leave and becomes leaving afterwards has to wait with actually executing \leave{} until that node has left (i.e., was replaced by a replacement).

One may ask how a leaving node $v$ can determine that it has received and successfully forwarded all messages sent to it to $v'$.
Therefore, we additionally assume that for each message sent via an edge in the system, an acknowledgment is sent back to the sending node (except for acknowledgments, for obvious reasons).
Each node then stores, for each edge, the number of acknowledgments it is still waiting for.
We assume also that a node knows all other nodes with incoming connections to it (this can, e.g., be achieved in that each node that establishes a new edge first introduces itself and waits for an acknowledgment before it uses the edge for any other messages).
Then, $v$ can ask all nodes with incoming connections to inform $v$ once they have received all acknowledgments for messages sent to $v$.
Once $v$ has received all responses, it knows that it does not receive any more messages.
After forwarding the received messages to $v'$ and receiving all acknowledgments for those, it knows it is safe to leave.

A left, middle, or right virtual node $u$ that created a replacement $v'$ for its right neighbor $v$ is called the node \emph{responsible for} $v'$.
Note that $v'$ may receive an additional \leave request from a node $w$.
In this case, the process emulating $u$ would spawn an additional node $w'$ and everything is carried out as though $v'$ were a normal node.
However, we say that $u$ is also responsible for $w'$.
This way a left, middle, or right virtual node may become responsible for a number of nodes. 
Similar to the joining of nodes, a node $u$ responsible for at least another node sends an additional number $B.l \in \mathbb{N}_0$ in the batch $B$ it sends out next, representing the number of \leave requests that $u$ has become responsible for since it last sent out a batch.

The rest is analogous to the join case:
As soon as the number of leave requests falls below half of the number of nodes emulated, the anchor initiates the update phase during which each node $u$ responsible for a set of nodes $v_1, \dots, v_k$ deletes these nodes and updates the De Bruijn Graph accordingly.
Once all acknowledgments for this have been propagated up in the tree, the update phase is left again.
Note that both joins and leaves may be handled in the same update phase.

On a sidenote, one may ask what happens if a joining node $v$ joins at some node $w$ that is currently in the process of leaving.
While $w$ is alive and has edges to some non-leaving nodes, $w$ can forward $v$ such that $v$ stays in the system.
However, once $w$ has left the system and is not alive anymore, $v$ cannot join the system through $w$.
Still, $v$ can detect if $w$ is not active anymore and then try joining the system from another node.

\paragraph{Updating the Anchor}
When $v_0$ wants to leave, we proceed similar as for the join case: 
$pred(v_0)$ will become the node responsible for $v_0$ and perform the duties of the anchor and at the very end of the update phase, the anchor information is transferred to the node that then has the minimum identifier.

\section{Analysis} \label{sec:analysis}
To prove that \pname{} implements a distributed queue according to Definition~\ref{def:semantics}, we define a total order on the \enqueue{} and \dequeue requests.
To do so we specify an algorithm that assigns each request a unique value from $\mathbb{N}$:
 First, initialize a virtual counter $c$ at the anchor with $1$ as its initial value (this value is transferred if the anchor is changed due to a join or a leave).
 We assign a virtual counter to each \enqueue{} or \dequeue request $op$ in the following way: 
 Recall that when $op$ is initiated, it causes the increase of an $op_i$ value of one batch $B$.
 Virtually assign to $value(op)$ the new value of $op_i$.
 We also say that $op$ \emph{belongs to $B$ at index $i$}.
 When $B$ is combined with another batch on its way up in the aggregation tree, choose one of the batches as the first one and one as the second one. 
 If $B$ is the second one, let $op_i'$ be the $i$-th entry of the other batch and add $op_i'$ to $value(op)$.
 In any case, $op$ belongs to the combined batch afterwards.
 Proceed in this way for every combination of batches up to the anchor.
 When the anchor processes the batch $(op_1'',\ldots,op_{k''}'')$ which $op$ belongs to, add $c + \sum_{j=1}^{i-1}op_j$ to $value(op)$.
 Afterwards, the anchor updates $c$ by $\sum_{j=1}^{k''}op_j$.
 Intuitively, imagine the anchor would process every request individually:
 Then it would first consider all $op_1''$ \enqueue{} requests, then all $op_2''$ \dequeue requests, and so on.
 The final value of $op$ would then be the number of requests that the anchor has served up to (and including) $op$.
 
 Observe that the values are unique.
 In the following, let $\prec$ be the order defined by the values given this way.
 The following lemmas follow from the protocol description (check the way we assigned values to the requests and how the intervals are assigned to the requests):
 \begin{lemma}\label{lem:deqdeq}
  If, for two \dequeue requests \deq{u}{i}, \deq{v}{j} that get assigned positions, $pos_a, pos_b$, respectively, $\deq{u}{i} \prec \deq{v}{j}$, then $pos_a < pos_b$.
 \end{lemma}
 \begin{lemma}\label{lem:enqenq}
  If, for two \enqueue{} requests \enq{u}{i}, \enq{v}{j} that get assigned positions, $pos_a, pos_b$, respectively, $\enq{u}{i} \prec \enq{v}{j}$, then $pos_a < pos_b$.
 \end{lemma}
 \begin{lemma}\label{lem:deqenq_enqdeq}
   If a \dequeue request \deq{u}{i} gets assigned a position $pos_a$, then for every \enqueue{} request \enq{v}{j} with $value(\deq{u}{i}) < value(\enq{v}{j})$ the position $pos_b$ assigned to it satisfies $pos_b > pos_a$.
   Likewise, if an \enqueue{} request \enq{u}{i} gets assigned a position $pos_a$, then for every \dequeue request \deq{v}{j} with $value(\deq{v}{j}) < value(\enq{u}{i})$ the position $pos_b$ assigned to it satisfies $pos_b < pos_a$.
 \end{lemma}
 \begin{lemma}\label{lem:deq_sequence}
  Assume there is a sequence of \dequeue requests $deq_1, \dots, deq_k$ that belong to the same batch $B$ and the same index $l$ such that $value(deq_1) < \dots < value(deq_k)$.
  If $deq_i$ returns $\bot$ for some $i \in \{1, \dots, k\}$, then all $deq_j$ with $i < j \leq k$ also return $\bot$ and $[v_0.first, v_0.last]$ is empty after index $l$ of batch $B$ has been processed in $v_0$ in Stage~2. 
 \end{lemma}
This lemma directly implies: 
 \begin{corollary}\label{cor:deqenq}
  If a \dequeue request \deq{u}{i} returns $\bot$ then every \deq{v}{j} request with $value(\deq{u}{i}) < value(\deq{v}{j})$ does not return an element $e$ added by an \enqueue{} request \enq{w}{k} with $value(\enq{w}{k}) < value(\deq{u}{i})$.
 \end{corollary}
\begin{lemma}\label{lem:prev_deq_always_exists}
 If a \dequeue request \deq{u}{i} gets assigned a position $pos$, then for every \enqueue{e} request that received a position $pos' < pos$ there exists a \deq{v}{j} request with $value(\deq{v}{j}) < value(\deq{u}{i})$ that returns $e$.
\end{lemma}
The reason is that the dequeue intervals always start with the lowest possible value.

\begin{lemma}\label{lem:get_always_answered}
 Every \textsc{Get} operation issued by any of the nodes is answered in finite time.
\end{lemma}
\paragraph{Proof sketch}
Note that the way \pname{} deals with leave requests makes sure that no messages get lost during leave as was argued in Section~\ref{sec:join_leave}.
Furthermore, check in the protocol description that whenever a \textsc{Get} message is at a node $u$ that is not responsible for storing the position corresponding with the \textsc{Get} message, $u$ knows a node that is closer to the node responsible for storing the position.
Thus, each \textsc{Get} message will eventually reach the node that is responsible for storing it (note that even if the node responsible for storing it changes meanwhile, then the old node responsible for storing it knows the new one and can forward the message accordingly).
If that node already stores the element required by the \textsc{Get} message, it can be answered directly.
Otherwise, check that the same we said about the \textsc{Get} message analogously applies to the corresponding \textsc{Put} message.
Thus, the element will eventually arrive at that node and the \textsc{Get} message can be answered.
 
We are now ready to prove the following theorem:
\begin{theorem}
	\pname{} implements a data structure that is sequentially consistent.
\end{theorem}
\begin{proof}
First of all note that due to the protocol description and Lemma~\ref{lem:get_always_answered}, every \dequeue request returns a value (i.e., either $\bot$ or some element $e \in \mathcal E$).
 We will consider all four requirements of Definition~\ref{def:semantics} individually.
 
 First, consider an arbitrary \dequeue request \deq{w}{j} that returns a value $e \in \mathcal E$ that was added due to an \enqueue{} request \enq{v}{i}.
 Since the position in the DHT is the same for both these requests, Lemma~\ref{lem:deqenq_enqdeq} implies that $value(\enq{v}{i}) < value(\deq{w}{j})$.
 
 Second, again consider an arbitrary \dequeue request \deq{w}{j} that returns a value $e \in \mathcal E$ that was added due to an \enqueue{} request \enq{v}{i}.
 For the first part, assume for contradiction that there is a \deq{u}{k} that returns $\bot$ with $value(\enq{v}{i}) < value(\deq{u}{k}) < value(\deq{w}{j})$.
 Then, Corollary~\ref{cor:deqenq} implies that \deq{w}{j} cannot return $e$, which is a contradiction.
 For the second part, assume for contradiction that there is an \enq{u}{k} whose element $e' \in \mathcal E$ is never returned with $value(\enq{u}{k}) < value(\enq{v}{i}) < value(\deq{w}{j})$.
 Combining Lemma~\ref{lem:enqenq} with Lemma~\ref{lem:prev_deq_always_exists} yields the desired contradiction also here.
 
 For the third requirement, consider an arbitrary \dequeue request \deq{v}{j} that returns a value $e \in \mathcal E$ that was added due to an \enqueue{} request \enq{u}{i} and an arbitrary \dequeue request \deq{x}{l} that returns a value $e' \in \mathcal E$ that was added due to an \enqueue{} request \enq{w}{k}.
 For the first part, assume for contradiction that $value(\enq{u}{i}) < value(\enq{w}{k}) < value(\deq{x}{l}) < value(\deq{v}{j})$.
 Lemma~\ref{lem:enqenq} yields that for the positions $pos_a$ and $pos_b$ assigned to \enq{u}{i} and \enq{w}{k}, respectively, $pos_a < pos_b$ holds.
 Note that $pos_a$ is assigned to \deq{v}{j} and $pos_b$ is assigned to \deq{x}{l}.
 However, Lemma~\ref{lem:deqdeq} would imply $pos_b < pos_a$, which yields a contradiction.
 The second part of the third requirement is analogous.
 
 The fourth requirement is directly satisfied by the way we defined $\prec$.
 This completes the proof of the theorem.
\end{proof}

In the following, we want to analyze the runtime of the operations \enqueue{}, \dequeue, \join{} and \leave.
We start with \enqueue{} and \dequeue requests.

\begin{theorem} \label{theorem:runtime}
	Each request \enqueue{} or \dequeue needs $\mathcal O(\log n)$ rounds w.h.p. until it is processed correctly on the distributed queue.
\end{theorem}

\begin{proof}
	Consider an arbitrary request $op \in \{\enqueue{}, \dequeue\}$.
	Assume an $op$ is generated by some node $v \in V$.
	By Corollary~\ref{cor:tree_height} we need $\log n$ rounds w.h.p. to transfer $op$ to the anchor node $v_0$ (Stage 1) as part of a batch.
	Thus it takes $\log n$ rounds w.h.p. to assign a position to each request (Stages~2 and 3).
	Finding the corresponding node $u$ for the position in the DHT and transferring the \textsc{Put}/\textsc{Get} operation for $op$ takes again $\log n$ rounds w.h.p. by Lemma~\ref{lemma:LDB:routing}.
	Note that if $op = \dequeue$, then we only have a constant message overhead for \textsc{Get}, as $u$ is able to send the result of \textsc{Get} to $v$ in one round.
	Summing it all up, we need $\mathcal O(\log n)$ number of rounds w.h.p.
\end{proof}

We obtain the following corollary, which shows that our approach is indeed scalable for a large number of incoming requests.

\begin{corollary} \label{cor:scalable_op}
	Assume a node $v \in V$ has stored an arbitrary amount of queue requests in $v.W$.
	The number of rounds, needed to process all requests successfully is $\mathcal O(\log n)$ w.h.p.
\end{corollary}

\begin{proof}
	Follows from Theorem~\ref{theorem:runtime} and the fact that we process requests in batches.
\end{proof}

Corollary~\ref{cor:scalable_op} emphasizes the advantages of processing multiple requests at once via batches: Imagine a node $v$ that generates one queue request in each round.
If a single queue request $op$ takes $\mathcal O(\log n)$ rounds to finish and $v$ is prohibited to process any further request before $op$ is finished, $v$'s local storage would eventually overflow.
For \pname{} however, $v$ is able to flush all requests contained in $v.B$ after $\mathcal O(\log n)$ rounds w.h.p.

\begin{theorem}
	Assume that at the beginning of the update phase there are $n$ joining nodes ($n/2$ node replacements).
	Then the update phase finishes after $\mathcal O(\log n)$ rounds w.h.p., if no node wants to join/leave the system in the meantime.
\end{theorem}

\begin{proof}
	By Corollary~\ref{cor:tree_height}, we need $\mathcal O(\log n)$ rounds w.h.p. to propagate the start of the update phase to all nodes in the aggregation tree.
	It is easy to see that a node $v$ responsible for multiple \join{} or \leave requests can process these requests in a constant amount of rounds once it got the permission from $p(v)$ in the update phase.
	The only case that may exceed the claimed upper bound is the case where the (old) anchor transfers its data to the new anchor, i.e., to the node with minimal label.
	However, with $n$ nodes joining, each old node is only responsible for at most $\mathcal O(\log n)$ joining nodes w.h.p.
	This implies that there are only $\mathcal O(\log n)$ joining nodes w.h.p. with smaller label than the anchor.
	The same argumentation holds for leaving nodes.
\end{proof}

Now we want to analyze the size of messages that are sent over communication channels in the network.
Obviously, the messages containing the most data are the ones containing a batch.
Thus, we want to get an upper bound on the maximum batch size.

\begin{theorem} \label{theorem:batch_size}
	Batches representing \enqueue{}, \dequeue, \join{} and \leave requests have size $\mathcal O(\log n)$ w.h.p. if each node generates one such request per round.
\end{theorem}

\begin{proof}
	Note that \join{} and \leave requests a node $v$ is responsible for are represented in a batch by a single constant.
	By Theorem~\ref{theorem:runtime}, each batch $v.W$ containing \enqueue{}, \dequeue, \join{} and \leave requests needs $\mathcal O(\log n)$ rounds w.h.p. until it is processed.
	Therefore a batch $v.W$ of some node $v$ can only have a size up to $\mathcal O(\log n)$ w.h.p. until it is sent out via $v.B$ assuming each node generates one request per round: The size of the batch increases only if the type (\enqueue{} or \dequeue) of the  request generated in round $s_i$ differs from the type of the request generated in round $s_{i-1}$.
\end{proof}

Finally we note that \pname{} is fair regarding the number of elements that each node has to store.
This immediately follows from the fairness property of the DHT (Lemma~\ref{lemma:dht:fairness}) and the fact that each joining or leaving node gets or transfers its DHT data.

\begin{corollary}
	\pname{} is fair.
\end{corollary}

\section{Distributed Stack} \label{sec:stack}
In this section we propose some simple modifications to \pname{} in order to realize a scalable distributed stack that fulfills sequential consistency.
Instead of \enqueue{} and \dequeue requests, the stack provides requests \push{} and  \pop such that for a single process it resembles a LIFO data structure.
Definition~\ref{def:semantics} can then be adjusted easily.

A natural approach would be to just change the way in which the anchor computes the position intervals for \dequeue requests (see Stage~2 in Section~\ref{sec:enq_deq}).
Recall that the anchor computes the interval $[v_0.first, \min\{v_0.first + op_i - 1, v_0.last\}]$ in case there are $op_i$ consecutive \dequeue requests.
For $op_i$ consecutive \pop requests, we want the anchor to return the interval $[\max\{1, v_0.last - op_i + 1\}, v_0.last]$ and update $v_0.last$ to $\max\{0, v_0.last - op_i\}$ afterwards.
Observe that we do not need the variable $v_0.first$ anymore.
Processes decomposing their position intervals in stage $3$ now have to take out the maximum position in the interval first.
Unfortunately, this modification does not suffice on its own, because the assigned positions for inserted elements are not unique: For the operation sequence $(\push{x}, \pop, \push{y})$ both \push{} requests are assigned to the same position by the anchor, leading to elements being replaced in the DHT.
Therefore we have to make sure that the key under which elements are inserted into the DHT is unique: We introduce a variable $v_0.ticket \in \mathbb{N}$ at the anchor, which is increased by $i$ every time $v_0.last$ is increased by $i$, but is never decreased, i.e., $v_0.ticket$ is monotonically increasing.
Intuitively, $v_0.ticket$ represents the number of \push{} requests ever processed at the anchor, whereas $v_0.last$ represents the current size of the stack.
A request is now assigned a pair $(position, ticket) \in \mathbb{N} \times \mathbb{N}$ instead of just a single position.
For such a pair $(p, t)$ that got assigned to a \push{x} request, we store $(p, t)$ and $x$ at the node that is responsible for position $p$ in the DHT.
A \pop request that got assigned to the pair $(p', t')$ searches the DHT for the node $v$ that is responsible for position $p'$.
After arrival at $v$, we remove the element with ticket $t \leq t'$ from $v$ and return it  to the initiator of the \pop request.

Nodes are able to locally combine generated requests in order to answer them immediately: For instance, if node $v$ generates $k$ \push{} requests $p_1,\ldots, p_k$ followed by $k$ \pop requests $po_1,\ldots, po_k$, then $v$ can process all of these requests immediately by assigning the $k-i+1$-th \push{} request to the $i$-th \pop request for all $i \in \{1,\ldots,k\}$.
This is particularly advantageous in scenarios where the rate at which nodes generate requests is very high.
It is easy to see that we do not violate sequential consistency with this modification.
Furthermore it follows that all batches which are sent upwards the aggregation tree have the form $B = (op_1, op_2)$ with $op_1 \in \mathbb{N}$ representing \pop operations and $op_2$ representing \push{} operations.
This immediately yields the following theorem on the size of a batch:

\begin{theorem} \label{theorem:stack:batch:size}
	Batches representing \push{} and \pop, requests have constant size.
\end{theorem}

In contrast to Theorem~\ref{theorem:batch_size}, Theorem~\ref{theorem:stack:batch:size} holds for any rate in which nodes generate stack requests.

Since we consider the asynchronous message passing model, all that is left is to prevent the following scenario from happening: Consider the operation sequence $(a, b, c, d)$ with $a = \push{x}$, $b = \pop$, $c = \push{y}$ and $d = \pop$.
Then the anchor assigns the pair $(p,t)$ to $a$, $(p, t)$ to $b$, $(p, t+1)$ to $c$ and $(p, t+1)$ to $d$.
Due to asynchronicity in our system, the DHT requests representing $a, b, c$ and $d$ may arrive in the order $(a, d, c, b)$ at the node responsible for position $p$.
This leads to $d$ returning the element $x$, as the ticket value for $a$ is smaller than the ticket value for $d$.
Request $b$ does not find an element with ticket value smaller or equal than its own and consequently fails, violating sequential consistency.
In order to fix this, we force all nodes $v$ to wait in stage 4 before switching to stage 1 again, until all DHT-operations that $v$ has generated in stage 4 have been finished (we just have to add this constraint to the clause in lines 2-3 of Algorithm~\ref{algo:phase_1}).
Reconsidering the above example, it follows that the order of arrival of the DHT operations will be either $(a, b, c, d)$ or $(a, c, b, d)$, because $a$ and $d$ are guaranteed to be in different batches than $b$ and $c$ when combining requests as described above.
It is easy to see that both cases prevail sequential consistency.
We obtain the main result of this section:

\begin{theorem}
	The modified \pname{} protocol implements a stack that is sequentially consistent.
\end{theorem}

\join{} and \leave requests are processed in the exact same manner on the stack as described in Section~\ref{sec:join_leave}.
\section{Evaluation}
We implemented and evaluated \pname{} as well as its stack adaptation (see Section~\ref{sec:stack}) on different instances.
In this section we present and interpret the most important results of these experiments.

\subsection{Setup}
We implemented the protocols for the synchronous message passing model and performed the following experiment for instances up to $100000$ nodes: At the beginning of each (synchronous) round, we generate 10 queue requests and assign them to random nodes in the system.
After 1000 rounds we stop the generation of requests and wait until all requests that are still being processed have finished successfully.
For each finished request we measure the number of rounds it took the requests to finish. 
For the results presented in this section we always consider the average amount of rounds per requests.
We tested instances with different ratios of \enqueue{}/\dequeue requests, respectively, \push{}/\pop requests.

\subsection{Distributed Queue}
Consider Figure~\ref{fig:queue_results} for results on the distributed queue.

\begin{figure}[ht]
 	\centering
	\begin{tikzpicture}
		\begin{axis}[
   	 		xlabel={n},
   	 		ylabel={(avg.) \#rounds per request},
      		scaled x ticks = false,
          	xmin=0, xmax=100000,
    		ymin=0, ymax=150,
    		xtick={10000, 50000, 100000},
    		ytick={0,50, 100, 150},
    		legend pos=south east,
    		ymajorgrids=true,
    		grid style=dashed,
		]
			\addplot[
    			color=blue,
    			mark=square,
    		]
    		coordinates {
    			(50,40.6564)(100,48.2873)(500,63.8632)(1000,71.542)(5000,88.1265)(10000,95.7975)(50000,119.7074)(100000,127.5005)
    		};
			\addlegendentry{$1.0$}	
			\addplot[
    			color=red,
    			mark=square,
    		]
    		coordinates {
    			(50,40.7641)(100,48.5138)(500,63.9932)(1000,71.8012)(5000,88.369)(10000,95.9759)(50000,119.4166)(100000,127.4895)
    		};
			\addlegendentry{$0.75$}
			\addplot[
    			color=yellow,
    			mark=square,
    		]
    		coordinates {
    			(50,40.9518)(100,48.6158)(500,64.4369)(1000,72.1492)(5000,88.4588)(10000,96.1679)(50000,120.5234)(100000,128.2695)
    		};
			\addlegendentry{$0.5$}
			\addplot[
    			color=green,
    			mark=square,
    		]
    		coordinates {
    			(50,35.6329)(100,42.5268)(500,55.3101)(1000,62.5008)(5000,77.6353)(10000,84.8457)(50000,105.4715)(100000,111.3723)
    		};
			\addlegendentry{$0.25$}
			\addplot[
    			color=orange,
    			mark=square,
    		]
    		coordinates {
    			(50,29.6391)(100,35.4468)(500,46.5561)(1000,52.4935)(5000,66.5815)(10000,72.3194)(50000,89.7779)(100000,95.3227)
    		};
			\addlegendentry{$0.0$}
		\end{axis}
	\end{tikzpicture}
 	\caption{Average number of (synchronous) rounds per request on the distributed queue. The graphs represent the different probabilities $p$ that a generated request is an \enqueue{} operation, meaning that $1-p$ is the probability that a generated request is a \dequeue operation.}
 	\label{fig:queue_results}
\end{figure}
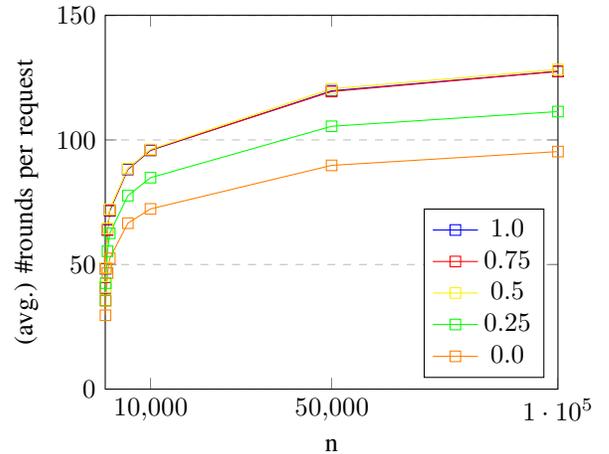

One can see that the average number of rounds for a request to finish scales logarithmically in the number of nodes $n$.
As soon as the \enqueue{} rate drops below $0.5$ the queue performs better, because the queue is empty most of the times.
This implies that \dequeue operations do not have to search for a position in the DHT, as they can be processed immediately as soon as the requesting node receives the position intervals from the anchor.
Interestingly, the curves for \enqueue{} rates of $0.5$ or higher are almost the same, which means that \dequeue operations waiting for the corresponding \enqueue{} operations in the DHT do not have a significant impact on the performance.

Roughly, these curves correspond to $3$ times the height of the aggregation tree (denoted as $ATH \approx \log n$) plus the average number of rounds it takes for a DHT operation to finish: A queue request first has to wait after generation until the next aggregation phase begins (on average $ATH$ rounds), then it is aggregated to the root ($ATH$ rounds) and assigned a position ($ATH$ rounds).
Afterwards we process the corresponding DHT operation in approximately $\log n$ rounds.

\subsection{Distributed Stack}
Consider Figure~\ref{fig:stack:results} for results on the distributed stack.

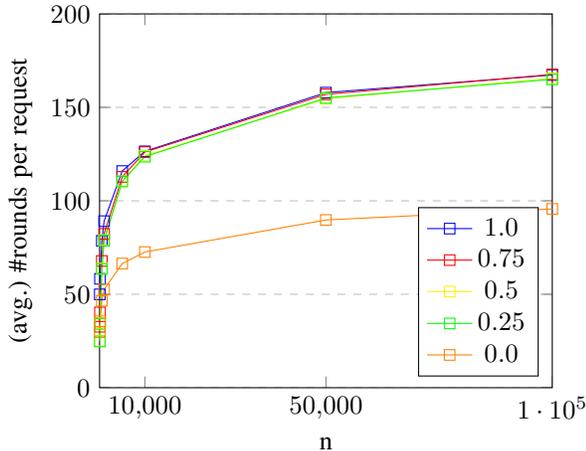
\begin{figure}[ht]
 	\centering
	\begin{tikzpicture}
		\begin{axis}[
   	 		xlabel={n},
   	 		ylabel={(avg.) \#rounds per request},
      		scaled x ticks = false,
    		xmin=0, xmax=100000,
    		ymin=0, ymax=200,
    		xtick={10000, 50000, 100000},
    		ytick={0,50, 100, 150, 200},
    		legend pos=south east,
    		ymajorgrids=true,
    		grid style=dashed,
		]
			\addplot[
    			color=blue,
    			mark=square,
    		]
    		coordinates {
    			(50,49.9538)(100,58.2514)(500,78.6243)(1000,89.1162)(5000,115.9221)(10000,126.4649)(50000,157.8873)(100000,167.3154)
    		};
			\addlegendentry{$1.0$}	
			\addplot[
    			color=red,
    			mark=square,
    		]
    		coordinates {
    			(50,32.6301)(100,40.2965)(500,67.8188)(1000,82.3171)(5000,112.809)(10000,126.0784)(50000,156.9622)(100000,167.6223)
    		};
			\addlegendentry{$0.75$}
			\addplot[
    			color=yellow,
    			mark=square,
    		]
    		coordinates {
    			(50,25.1101)(100,34.5587)(500,63.7763)(1000,78.993)(5000,110.4514)(10000,124.0818)(50000,154.5981)(100000,165.598)
    		};
			\addlegendentry{$0.5$}
			\addplot[
    			color=green,
    			mark=square,
    		]
    		coordinates {
    			(50,24.7742)(100,34.1142)(500,63.5453)(1000,78.6899)(5000,110.3842)(10000,123.602)(50000,155.0656)(100000,164.9926)
    		};
			\addlegendentry{$0.25$}
			\addplot[
    			color=orange,
    			mark=square,
    		]
    		coordinates {
    			(50,29.7331)(100,35.5132)(500,46.4788)(1000,52.6014)(5000,66.4461)(10000,72.5963)(50000,89.7435)(100000,95.659)
    		};
			\addlegendentry{$0.0$}
		\end{axis}
	\end{tikzpicture}
 	\caption{Average number of (synchronous) rounds per request on the distributed stack. The graphs represent the different probabilities $p$ that a generated request is a \push{} operation, meaning that $1-p$ is the probability that a generated request is a \pop operation.}
 	\label{fig:stack:results}
\end{figure}

Same as for the queue, the average number of rounds for a request scales logarithmically in the number of nodes $n$.
However, the stack performs a bit slower than the queue, because we wait at the end of stage 4 until all DHT operations have finished.
This delays the start of the next aggregation phase and leads to all curves representing \push{} ratios greater than $0$ being roughly the same.
Obviously the stack performs better if we only generate \pop operations.
In fact, the curve for a \push{} ratio of $0$ is the same as the corresponding curve for the queue, which makes sense, since both data structures do not have to issue any DHT operations.

Unfortunately, we cannot see the impact of the local combination of operations in this setting, because the probability that more than one operation is generated at a node $v$ in the same aggregation phase is very low.
Therefore we perform an additional experiment: We consider an instance of $n = 10000$ nodes and generate requests at nodes with constant probability $p \in \{0.05, 0.1, 0.15, 0.2, 0.25, 0.5, 1\}$ at each round. For instance, if $p = 1$, we generate one request at each node in each round leading to $1000n = 10^7$ generated requests after $1000$ rounds.
The probability that a generated request is an \enqueue{}/\push{} operation is $0.5$.
Again we looked at the average number of rounds it takes a request to be processed successfully for both, the queue and the stack.
The results can be seen in Figure~\ref{fig:high_load_results} (note that the horizontal axis now represents the different probabilities $p$ mentioned above).

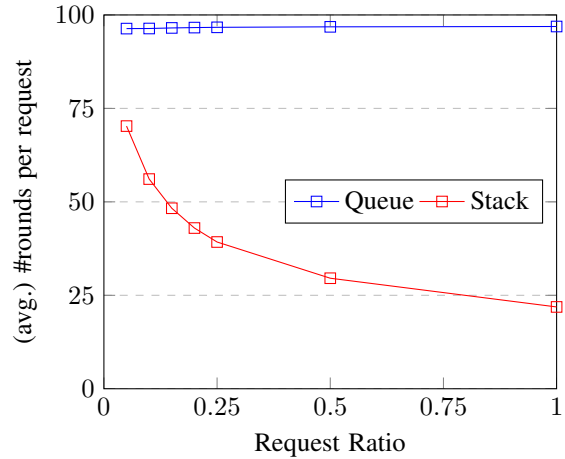
\begin{figure}[ht]
 	\centering
	\begin{tikzpicture}
		\begin{axis}[
   	 		xlabel={Request Ratio},
   	 		ylabel={(avg.) \#rounds per request},
    		xmin=0, xmax=1.0,
    		ymin=0, ymax=100,
    		xtick={0, 0.25, 0.5, 0.75, 1},
    		ytick={0,25, 50, 75, 100},
  legend style={
    at={(0.685,0.56)},
    anchor=north,
    legend columns=-1
  },     	ymajorgrids=true,
    		grid style=dashed,
		]
			\addplot[
    			color=blue,
    			mark=square,
    		]
    		coordinates {
    			(0.05,96.3482757404242)(0.1,96.372055558966)(0.15,96.5245560762212)(0.2,96.6191124058358)(0.25,96.6885722548602)(0.5,96.8187467915529)(1,96.9081668331668)
    		};
			\addlegendentry{Queue}	
			\addplot[
    			color=red,
    			mark=square,
    		]
    		coordinates {
    			(0.05,70.2688717957758)(0.1,56.0945963581573)(0.15,48.2884111981324)(0.2,42.9928458873924)(0.25,39.2564344852411)(0.5,29.5688440499032)(1,21.887537062937)
    		};
			\addlegendentry{Stack}
		\end{axis}
	\end{tikzpicture}
 	\caption{Average number of (synchronous) rounds per request on the queue/stack with different request ratios and $n = 10000$.}
 	\label{fig:high_load_results}
\end{figure}

Here we can see that the stack's performance gets even better if the rate at which requests are generated increases.
This is due to nodes issuing multiple requests in the same aggregation phase, which leads to the stack being able to combine operations locally, such that they can be processed immediately.

\section{Conclusion} \label{sec:conclusion}
We presented the protocol \pname{} for a distributed queue that guarantees sequential consistency and is able to process requests fast even for a high rate of incoming requests.

A challenging task would be to make \pname{} self-stabilizing, such that the network can recover itself from faulty states.
However, due to the various amount of variables that have to be stored at each node and the fact that we are in an asynchronous environment, one will quickly have to weaken the queue semantics.

\bibliography{literature}
\bibliographystyle{IEEEtran}
\end{document}